\documentclass[letterpaper, 10 pt, conference]{ieeeconf}

\IEEEoverridecommandlockouts

\usepackage{amssymb,amsmath}

\usepackage{amsthm}
\usepackage{mathtools}
\usepackage{multirow}
\usepackage{graphicx}
\usepackage{comment}
\usepackage[sort,compress,noadjust]{cite}
\usepackage[font=footnotesize]{subcaption}
\usepackage[font=footnotesize]{caption}

\usepackage{amsfonts}
 \usepackage{tabularx}
\usepackage{graphicx}
\usepackage{array}
\usepackage{float}
\usepackage{hhline}
\usepackage{algorithmic}
\usepackage[linesnumbered, ruled, vlined]{algorithm2e}
\usepackage[colorlinks,urlcolor=blue]{hyperref}
\usepackage{setspace}
\usepackage{colortbl}
\usepackage{arydshln}
\SetCommentSty{mycommfont}

\DeclareMathOperator*{\argmin}{arg\,min}

\usepackage{accents}

\usepackage{placeins}

\theoremstyle{plain}
\newtheorem{theorem}{Theorem}
\theoremstyle{plain}
\newtheorem{proposition}{Proposition}
\theoremstyle{plain}

\theoremstyle{plain}
\newtheorem{corollary}{Corollary}
\theoremstyle{definition}
\newtheorem{assumption}{Assumption}
\theoremstyle{remark}
\newtheorem{remark}{Remark}

\newtheoremstyle{definition}{}{}{}{}{\bfseries}{.}{.5em}{\thmname{#1}\thmnumber{ #2}\thmnote{ (#3)}}
\theoremstyle{definition}
\newtheorem{definition}{Definition}

\usepackage[capitalize]{cleveref}
\crefformat{equation}{(#2#1#3)}
\Crefformat{equation}{Equation~(#2#1#3)}
\Crefname{equation}{Equation}{Eqs.}

\title{\LARGE\bf Information Control Barrier Functions: \\ Preventing Localization Failures in Mobile Systems Through Control}

\author{Samuel G. Gessow, David Thorne, and Brett T. Lopez%
\thanks{
All authors are with VECTR Laboratory, University of California Los Angeles, Los Angeles, CA, USA , {\tt\footnotesize \{sgessow, davidthorne, btlopez\}@ucla.edu}}
}

\begin{document}
\maketitle
\thispagestyle{empty}
\pagestyle{empty}


\begin{abstract}
This paper develops a new framework for preventing localization failures in mobile systems that must estimate their state using measurements. 
Safety is guaranteed by imposing the nonlinear least squares optimization solved in modern localization algorithms remains well-conditioned.
Specifically, the eigenvalues of the Hessian matrix are made to be always positive via two methods that leverage control barrier functions to achieve safe set invariance. 
The proposed method is not constrained to any specific measurement or system type, offering a very general solution to the safe mobility with localization problem.
The efficacy of the approach is demonstrated on a system being provided range-only and heading-only measurements for localization.
\end{abstract}

\section{Introduction}
\label{sec:introduction}
Localization involves acquiring and processing measurements in real-time to determine the state of a mobile platform in a given frame of reference.
Localization often dictates the overall performance of mobile systems because control, planning, and autonomy heavily rely on accurate state estimates. 
As a result, the field of localization has seen considerable theoretical and practical advances over the past decade, especially as sensors and processors have become lighter, smaller, and more powerful. 
However, despite major progress, localization failures---the inability to determine a system's state---still occur and are usually catastrophic when there are decision making algorithms downstream of the localization stack.
Since localization failures can occur rapidly, there is little to no time for corrective action to restore localization so either i) the localization algorithm must be guaranteed to never fail no matter the environment or ii) the system {must \emph{predict} when a failure is about to occur} and take preemptive corrective actions. 
Fundamentally, odometry algorithms rely on geometric or visual features distributed throughout the environment \cite{barfoot2024state}, meaning the performance of these algorithms is inherently state-dependent, and this state-dependency is ignored as the system is moving through the environment leading to localization failures. 
This work addresses the need for \emph{safety in online localization} for mobile systems by developing a new framework that leverages the simplicity and effectiveness of control barrier functions (CBFs) in conjunction with matrix calculus and nonlinear least squares to ensure a system never enters a state where localization may fail.

Prior works in planning have investigated designing trajectories through ``measurement rich" parts of the environment that improve the observability of certain states/model parameters.
Belief-space planning \cite{prentice2009belief,platt2010belief} solves (approximately) a partially observable Markov decision making process that must balance information gain about system states with reaching a final goal location.
Perception-aware planning \cite{costante2016perception,falanga2018pampc} frameworks for visual-inertial odometry design trajectories that require sufficient visual features stay in the field-of-view such that localization is well-constrained.
Rather than attempting to reduce estimation covariances, observability-aware planning \cite{hausman2017observability} designs motions that maximize the Observability Gramian for specific states or model parameters.
One limitation of the aforementioned works is the reliance on \emph{a priori} information in the form of a global or local map of the environment which is unrealistic in many applications. 
Another is that localization quality is treated as a soft constraint when making control or planning decisions so safety cannot be formally guaranteed.
More recently, CBFs have been proposed in specific instances of this problem such as visual servoing \cite{salehi2021constrained} and in conjunction with the observability matrix \cite{coleman2024using}.
However, these methods are tailored to specific applications and do not provide a general approach for ensuring accurate localization across a broad class of mobile platforms.

This letter introduces a novel framework for achieving safe mobility with online localization for nonlinear systems through safety-critical control.
Specifically, we formulate a set invariance condition via CBFs that imposes a minimum threshold on the eigenvalue of the Hessian matrix associated with the {optimization} solved in modern localization algorithms.
In other words, our formulation prevents online nonlinear least squares from becoming ill-conditioned by taking preventative actions via control. 
Importantly, the method does not rely on proxies for localization performance but directly guarantees sufficient conditions for safe operation.
It achieves this without requiring \textit{a priori} information about the environment, so it can be used in unknown environments and be immediately combined with an existing trajectory planner.
Our framework leverages the state-dependency of the minimum eigenvalue of the Hessian to achieve safety so it is necessary to guarantee that the minimum eigenvalue is sufficiently smooth.
To ensure this, we present two methods: one that ensures the eigenvalues are always differentiable even in the non-simple, i.e., repeated, case and another that forces the eigenvalues to always remain simple. 
Our approach can be applied to systems of high relative degree in addition to situations where avoiding detection is the primary objective rather than localization performance.
The proposed approach is validated on a double integrator system using range-only and heading-only localization with beacons.

\textit{Notation:} 
The set of strictly-positive scalars is denoted $\mathbb{R}_{+}$.
The set of real symmetric matrices of dimension $n$ is denoted $\mathcal{S}^n$, with the set of positive definite matrices $\mathcal{S}^n_+$.
The eigenvalues of matrix $A$ are $\lambda(A)$ and the smallest eigenvalue is $\lambda_{\mathrm{min}} (A) = \min \, \lambda(A)$. 
A matrix that is a function of $x$ and $t$ is denoted as $A(x,t)$ and is smooth if all of its elements $a_{ij}(x,t)$ are smooth functions of $x$ and $t$. 
The Lie derivative of vector field $x \mapsto h(x)$ along the flow of another vector field $x \mapsto f(x)$ is denoted $L_f h(x)= \nabla h(x)^\top f(x)$. 
The norm of vector $x\in \mathbb{R}^n$ with $\Sigma \in \mathcal{S}_+^{n}$ is $\|x\|^2_{\Sigma} = x^\top \Sigma^{-1} x$.
 
\section{Problem Formulation and Preliminaries}
\label{sec:problem}

\subsection{Overview}
Consider the nonlinear system
\begin{equation}
    \begin{aligned}
    \dot x & = f(x) + g(x) u \\
    y & = \mathfrak{m}(x),
    \end{aligned}
    \label{eq:eqn_sys}
\end{equation}
with state $x \in \mathbb{R}^n$, control input $u \in \mathbb{R}^m$, measured output $y \in \mathbb{R}^p$, dynamics $f: \mathbb{R}^n \rightarrow \mathbb{R}^{n}$, control input matrix $g: \mathbb{R}^n \rightarrow \mathbb{R}^{n\times m}$, and nonlinear measurement model $\mathfrak{m} : \mathbb{R}^n \rightarrow \mathbb{R}^p$. 
There are several methods for generating an estimate of the state $x$ using output $y$, e.g., Kalman filtering and state observers.
However, recent works in the robotics literature have shown the superiority of optimization-based approaches that involve solving a nonlinear least squares problem in an online fashion \cite{leutenegger2015keyframe,barfoot2024state,chen2023dliom}. 
In particular, if measurements $m \in \mathbb{R}^p$ are available, then one can use numerical optimization to find the optimal state estimate $\hat{x}^*$ by solving
\begin{equation}
\label{eq:nl}
        \hat{x}^* = \argmin_{x \, \in \, \mathbb{R}^n} \,  \|m - \mathfrak{m}(x)\|_{\Sigma(x)}^2,
\end{equation}
where $\Sigma : \mathbb{R}^n \rightarrow \mathcal{S}^{n}$ is a state-dependent positive definite matrix that represents the level of confidence in $m$.
We allow $\Sigma$ to be state-dependent as a way to capture measurement degradation/dropout; this will be discussed more in \cref{sec:example}.
If the measurement model does not contain enough information to estimate the full state directly by solving \cref{eq:nl}, then the nonlinear least squares optimization can be performed over a window of measurements with a motion model. 
In either case, the proposed approach is applicable.
The manifestation of \cref{eq:nl} can be traced to nonlinear Maximum A Posteriori (MAP) estimation \cite{barfoot2024state} where maximizing the log likelihood function is more convenient when the likelihood function belongs to the exponential family, i.e., likelihood function is approximately normal.
The goal of this work is to ensure the optimization \cref{eq:nl} is well-constrained (to be defined later) so a unique state estimate is always available; a guarantee to be achieved through control.
In the sequel, we will denote the nonlinear least squares cost as $J(x,{m})$.

\subsection{Control Barrier Functions}
This work will employ controlled invariant sets to ensure the optimization \cref{eq:nl} always produces a unique estimate $\hat{x}^*$.
Specifically, set invariance will be achieved with {CBFs}.
The following definitions can be found in many previous works \cite{ames2016control,ames2019control} and are stated here for completeness. 

\begin{definition}
The set ${S}$ is {\emph{forward invariant}} if for every $x_{0} \in {S}, x(t) \in {S}$ for all $t$.
\end{definition}

\begin{definition}
The nominal system is \emph{safe} with respect to set ${S}$ if the set ${S}$ is forward invariant.
\end{definition}

\begin{definition}
A continuous function $\alpha : \mathbb{R} \rightarrow \mathbb{R}$ is an \emph{extended class $\mathcal{K}_{\infty}$ function} if it is strictly increasing, $\alpha(0)=0$, and is defined on the entire real line.
\end{definition}

\begin{definition}[cf.~\cite{ames2016control}]
Let $S \subset E \subset \mathbb{R}^n$ be the 0-superlevel set of a continuously differentiable function $h : E \rightarrow \mathbb{R}$ defined on the open and connected set $E \subset \mathbb{R}^n$. The function $h$ is a {\emph{control barrier function}} for \cref{eq:eqn_sys} if there exists an extended class $\mathcal{K}_\infty$ function $\alpha$ for all $x \in E$ such that
\begin{equation*}
\label{def:CBF}
    \sup_{u\in U} \left[ L_f h(x) + L_g h(x) u \right] \geq - \alpha(h(x)).
\end{equation*}
\end{definition}

\subsection{Other Useful Definitions}
This work will also make use of the following definitions from optimization theory, linear algebra, and analysis. 

\begin{definition}
\label{def:degeneracy}
    Let $F: \mathbb{R}^n \rightarrow \mathbb{R}$ be a twice differentiable function with critical point $z^*$ where the gradient of $F$ at $z^*$ vanishes, i.e., $\nabla F(z^*) = 0$. 
    The Hessian $H:\mathbb{R}^{n} \rightarrow \mathcal{S}^{n}$ of $F$ given by $H(z) = \nabla^2 F(z)$ is said to be \textit{degenerate} at the critical point $z^*$ if at least one eigenvalue of $H$ is zero.
\end{definition}

\begin{remark}
    An alternative to \cref{def:degeneracy} is that the Hessian matrix is degenerate if its determinant is zero. 
    This can be trivially shown to be equivalent to the eigenvalue condition.
\end{remark}

\begin{definition}
    An eigenvalue of a matrix is \textit{simple} if it has an algebraic multiplicity of one.
    In other words, a simple eigenvalue only appears once in the roots of the characteristic polynomial of the matrix under consideration.
\end{definition}

\begin{definition}[cf.~\cite{tao2006power}]
    Let $F: \mathbb{R} \rightarrow \mathbb{R}$. The function $F$ is \textit{analytic} on $\mathbb{R}$ if for all $a\in \mathbb{R}$ there exists an open interval $(a-r,a+r)$,  $r \in \mathbb{R}_{+}$ such that there exists a power series $\sum_{n=0}^{\infty}c_n(x-a)^n$ which has radius of convergence greater than or equal to $r$ and which converges to $F$ on $(a-r,a+r)$.
\end{definition}

\begin{remark}
    A function being analytic implies that it is infinitely differentiable, but the converse does not hold.
    This makes being analytic a stronger condition than being infinitely differentiable.
\end{remark}
\section{Main Results}
\label{sec:main_result}

\subsection{Overview}
In this section, we present a new approach that prevents the nonlinear least squares optimization defined in \cref{eq:nl} from becoming ill-conditioned through the use of {CBFs}.
Central to our approach is the following result that establishes a connection between the uniqueness of a critical point for a twice differentiable function and the eigenvalues of its corresponding Hessian. 

\begin{proposition}
\label{prop:unique}
    Let $z^* \in \mathbb{R}^n$ be a critical point of a twice differentiable function $F: \mathbb{R}^n \rightarrow \mathbb{R}$.
    If the eigenvalues of $H(z^*) = \nabla^2 F(z^*)$ are all positive locally then $z^*$ is a unique local minimum of $F$.
\end{proposition}

\begin{proof}
    Let $z^*$ be a critical point of $F$, then $\nabla F(z^*) = 0$ by definition.
    Consider a local perturbation around $z^*$, denoted as $\delta_z$, then $F(z^* + \delta_z) = F(z^*) + \nabla F(z^*)\delta_z + \delta_z^\top H(z^*) \delta_z + {\scriptstyle \mathcal{O}}(\|\delta_z\|^3)$.
    Since $\nabla F(z^*)=0$ and $H(z^*) \succ 0$ then $F(z^*+\delta_z) > F(z^*)$ so $z^*$ is a unique local minimum of $F$. 
\end{proof}

\noindent The following corollary follows directly from \cref{prop:unique}.

\begin{corollary}
\label{cor:deg}
    If the Hessian of a twice differentiable function is degenerate at a critical point then the critical point may not be unique locally.
\end{corollary}

\Cref{cor:deg} has an important implication for optimization-based localization algorithms: if the Hessian of \cref{eq:nl} is degenerate at a critical point, then the computed estimate $\hat{x}^*$ may not be unique.
From a localization perspective, a non-unique solution to \cref{eq:nl} is analogous to elements of the state vector being unobservable; a situation that can have devastating consequences for mobile systems that base decisions and control actions on estimates computed via \cref{eq:nl}.
Consider, e.g., \cref{fig:degeneracy_example}, which shows examples of cost functions that have a non-degenerate and degenerate Hessian in \cref{fig:non_degenerate} and \cref{fig:degenerate}, respectively.
The cost function shown in \cref{fig:degenerate} has infinitely many minima that are all valid solutions, any of which can be computed by the localization algorithm and used in feedback.
A controller using a non-unique estimate will behave erratically as the estimate is not an accurate representation of the system's true state, and would likely lead to any number of catastrophic events such state/actuator constraint violations or collisions with obstacles.

\begin{remark}
    It is important to note that a unique critical point (or estimate) might still be obtainable even when the Hessian is degenerate.
    In other words, a degenerate Hessian is not equivalent to the non-uniqueness of critical points.
    Nonetheless, requiring that the Hessian never becomes degenerate {is sufficient to guarantee} that a critical point is unique.
\end{remark}

\begin{figure}[t!]
\vspace{0.05in}
\centering
  \subfloat[Cost function with non-degenerate Hessian.]{\includegraphics[trim=110 20 170 50, clip, width=.45\columnwidth]{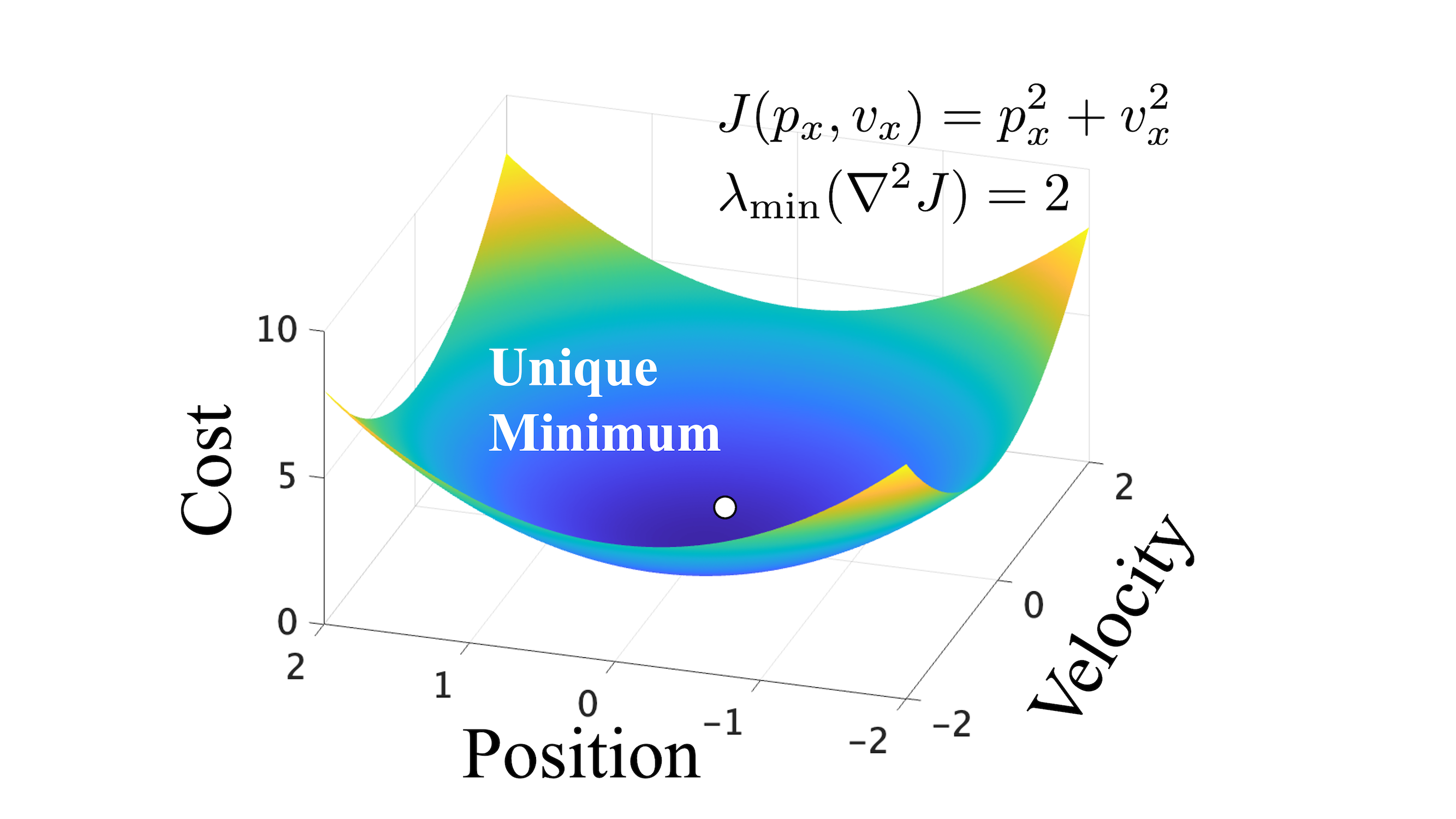}\label{fig:non_degenerate}}
  \hspace{1em}
  \subfloat[Cost function with degenerate Hessian.]{\includegraphics[trim=80 20 140 50, clip, width=.48\columnwidth]{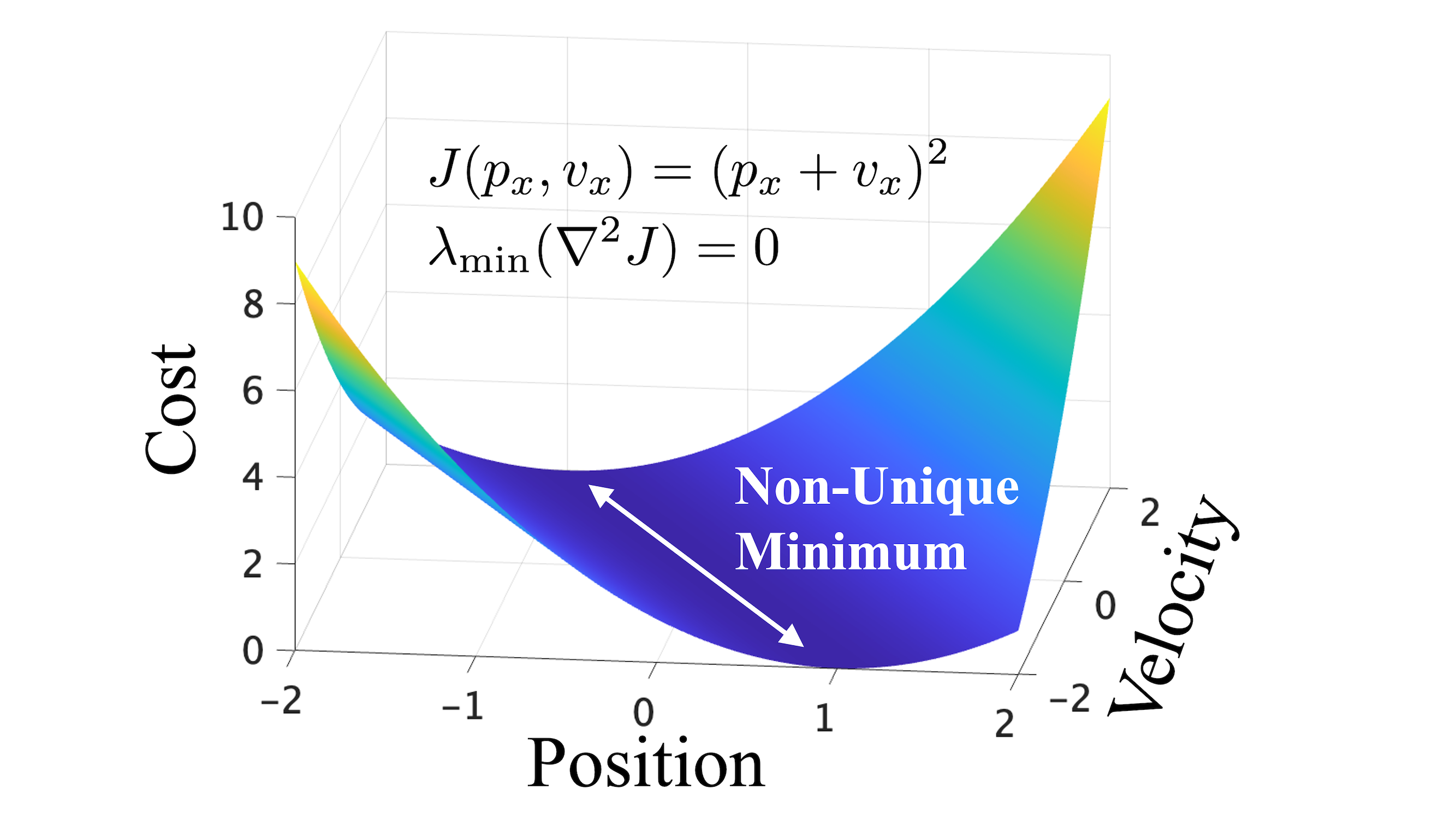}\label{fig:degenerate}}
  \caption{(a): Example cost function with non-degenerate Hessian that has a unique critical point. (b): Example cost function with degenerate Hessian that has a non-unique critical point. }
    \label{fig:degeneracy_example}
    \vskip -0.2in
\end{figure}

\subsection{Information Control Barrier Function}
Given the discussion above, we define the special class of {CBFs} that will be the focus for the remainder of this letter.

\begin{definition}
\label{def:icbf}
    Let $H : \mathbb{R}^n \times \mathbb{R}^{p} \rightarrow \mathcal{S}^n$ be the Hessian of the nonlinear least squares cost function \cref{eq:nl}. An {\emph{information control barrier function}} (I-CBF) is a control barrier function that renders the set ${S}=\{ x \in \mathbb{R}^n \, : \, \lambda_{\mathrm{min}} (H(x,m))\geq \lambda_{s}\}$ safe where $\lambda_{s} \in \mathbb{R}_+$.
\end{definition}

\begin{remark}
    The minimum eigenvalue of the Hessian and its derivatives are amenable to numerical computation if the Hessian of \cref{eq:nl} is a complicated function.
\end{remark}

\begin{remark}
\label{remark:detect}
    A noteworthy alternative to $S$ in \cref{def:icbf} is $S^c = \{ x \in \mathbb{R}^n \, : \, \lambda_{\mathrm{min}} (H(x,m))\leq \lambda_{s}\}$ which ensures the Hessian is always nearly degenerate. This can be used for detection avoidance applications (see \cref{sec:example}).
\end{remark}

The minimum eigenvalue of the Hessian of \cref{eq:nl} is a measure of how much information about the state is contained in the nonlinear least squares cost function. 
As such, the class of control barrier functions in \cref{def:icbf} ensures that localization information is sufficient for safe operation, i.e., the system's state can be discerned from the available measurements. 
A natural choice for a function $h$ to define the safe set in I-CBF  is 
\begin{align}
    \label{eq:eqn_h}
    h(x)&=\lambda_{\mathrm{min}}(H(x, m))-\lambda_{s}.
\end{align}

For \cref{eq:eqn_h} to be a valid CBF (based on \cref{def:CBF}), it must be continuously differentiable and so $\lambda_{\mathrm{min}}(H(x, m))$ must be continuously differentiable.
Handling the (potential) non-differentiability of the min operator is fairly straightforward, with approaches presented in \cite{glotfelter2017nonsmooth,molnar2023composing}. 
Establishing differentiability of the eigenvalues of the Hessian is more nuanced and requires a thorough investigation.
If other degeneracy measures were to be employed instead of \cref{eq:eqn_h}, e.g., the determinant or condition number of the Hessian, the question of differentiability of the eigenvalues of the Hessian must still be answered. 
Because of this, proving $\lambda(H)$ to be sufficiently smooth is integral to the proposed method.

The following proposition states an important result for matrices whose eigenvalues are simple, i.e., unique. 

\begin{proposition}
Let $A: \mathbb{R}^q \rightarrow \mathcal{S}^n$ be an $m$-times continuously differentiable real symmetric matrix with simple eigenvalues. Then, the eigenvalues are $m$-times continuously differentiable with respect to its argument.
\label{prop:simple_eigs}
\end{proposition}
\begin{proof}
Using Theorem 5.3 from \cite{serre2010matrices} there are analytic functions $\lambda(A)$ and $v(A)$ which are the eigenvalues and eigenvectors respectively. Given that $A$ is $m$-times continuously differentiable finishes the result. The key element in the proof in \cite{serre2010matrices} comes from the implicit function theorem, which requires simple eigenvalues.
\end{proof}

If the eigenvalues are non-simple (i.e. they cross\footnote{Interestingly, eigenvalue crossing has been a polarizing topic in quantum mechanics and other related fields, as the question of the so-called \emph{avoidance crossing} is still debated.}) then they still might be differentiable.
When the eigenvalues are non-simple, the following theorem can be used to guarantee differentiability but at the expense of stronger assumptions.

\begin{proposition}
     Let $A : \mathbb{R} \rightarrow \mathcal{S}^n$ be a real symmetric matrix which is analytic in its argument. Then, the eigenvalues are analytic with respect to its argument.
    \label{prop:analytic_eigs}
\end{proposition}
\begin{proof}
    The proof is a sub-case of one presented in \cite{kriegl2011denjoy1}.
\end{proof}

Given \cref{prop:simple_eigs,prop:analytic_eigs}, the Hessian of \cref{eq:nl} must now be studied to determine if either or when the aforementioned propositions are applicable. 
The first property that needs to be established is that, for the Hessian $H$ to be continuously differentiable in time, so too must $x$ and $m$.
Two assumptions are made to ensure $H$ is continuously differentiable.
\begin{assumption}
The nonlinear least squares cost function $J(x,m)$ is twice continuously differentiable in $x$.
\label{assump_cf}
\end{assumption}

\begin{assumption}
\label{assump_mm}
The state $x$ and measurement $m$ (and so to the measurement model $\mathfrak{m}$) are continuously differentiable.
\end{assumption}

With Assumptions~\ref{assump_cf} and \ref{assump_mm}, the Hessian $H$ is continuously differentiable in time but, as noted above, this alone is not enough to guarantee differentiability of the eigenvalues {everywhere}.
From \cref{prop:analytic_eigs}, if $H$ depends analytically on a single parameter then the eigenvalues of $H$ are differentiable even in the non-simple case.
These two cases lead to two different methods of constructing valid I-CBFs: one where $H$ 
is analytic and another where the eigenvalues of $H$ are always simple, i.e., the eigenvalues never cross.

\subsection{Analytic Hessian}
The goal of this section is to prove that the eigenvalues of the Hessian $H$ are sufficiently smooth.
From \cref{prop:analytic_eigs}, the eigenvalues of $H$ are analytic (and therefore smooth) provided that 
$H$ itself is an analytic function of only a \emph{single} parameter.
In our framework, this single parameter is time $t$, as the implementation of a controller renders all system dynamics solely time-dependent. 
To formalize this and prove differentiability for \cref{eq:eqn_sys}, we assume the following.

\begin{assumption} 
\label{assump_new}
The dynamics, measurement model, and control in \cref{eq:eqn_sys} are all analytic. 
\end{assumption}

\begin{theorem}
\label{thm:analytic}
    Consider the nonlinear system \cref{eq:eqn_sys} that satisfies Assumption~\ref{assump_new}. 
    If the nonlinear least squares cost function $J(x,m)$ in \cref{eq:nl} and measurements $m$ are sufficiently smooth then the eigenvalues of the Hessian $H(x,m) = \nabla^2 J(x,m)$ are analytic with respect to time $t$.
\end{theorem}
\begin{proof}
    The Cauchy-Kovalevskaya lemma guarantees that the solution to \cref{eq:eqn_sys} is analytic locally as the dynamics and control input are analytic by assumption.
    Since the solution to \cref{eq:eqn_sys} is a local analytic function of $t$, then $H(x,m) \approx H(x(t), \mathfrak{m}(x(t))) = H(t)$ locally. {Now that} $H$ is a function of a single parameter {$t$}, \cref{prop:analytic_eigs} applies, so the eigenvalues of the Hessian are guaranteed to be differentiable.
\end{proof}

\begin{remark}
It is possible that the eigenvalues of the Hessian are differentiable even without the assumptions made for \cref{thm:analytic} when working with specific types of systems. 
Future work will investigate characterizing the systems in which the ``avoidance of crossing" phenomenon occurs \cite{lax2007linear}.
\end{remark}

One of the stronger assumptions made in \cref{thm:analytic} is that the control input $u$ is analytic, as it is often computed via a quadratic program for safety-critical control.
While conditions under which quadratic programs yield analytic control are discussed in \cite{mestres2023regularity}, often quadratic programs produce only continuous control.
Cohen et al. \cite{cohen2023characterizing} explore cases where this issue can be mitigated by creating a smooth version of the controller that still satisfies the CBF condition but at the expense of possibly more control effort.
One solution discussed in \cite{cohen2023characterizing} involves modifying the closed-form solution to the quadratic program to make the controller analytic. 
If the quadratic program is of the form
\begin{equation}
    \begin{tabular}{rl}
        $\min_{u\in U}$ & $\ell(u,u_d)$\\
        subject to & $L_f h(x) + L_g h(x) u \leq  -\alpha(h(x))$,\\
    \end{tabular}
\label{eqn_QP}
\end{equation}
where $\ell(u,u_d)$ is convex in $u$, then it admits a closed form solution which can be found using KKT conditions \cite{boyd2004convex}. 
The closed form solution to \cref{eqn_QP} is, in general, not analytic as it is often piecewise.
For instance, when $\ell(u, u_d) = \tfrac{1}{2}\|u-u_d\|^2$, the closed-form solution for $u$ can be expressed as $u=u_d+\text{ReLu}(-\Psi)\frac{L_gh}{\|L_g h \|^2 }$ where $\Psi = L_fh + L_gh \, u_d + \alpha(h)$, which is not differentiable at $\Psi = 0$. 
However, an arbitrarily close analytic approximation of the ReLu function can be constructed while ensuring the CBF condition is satisfied.
For example, one analytic approximation of $u$ is
\begin{equation}
    u=u_d+\frac{1}{c}\ln(1+\exp(-c \Psi))\frac{L_gh}{\|L_g h\|^2},
    \label{eqn:smooth_relu}
\end{equation}
for $c > 1$ always satisfies the conditions for a CBF at the expense of conservatism.
The function $\alpha \in \mathcal{K}_{\infty}$ along with parameter $c$ can be adjusted to offset the conservatism.

Having shown the eigenvalues of the Hessian are analytic (and hence smooth) under the conditions needed for \cref{thm:analytic}, the only remaining item to address is the differentiability of the min operator in \cref{eq:eqn_h}.
We adopt the strategy in \cite{molnar2023composing} that utilizes a smooth under-approximation of the min operator to achieve differentiability without losing set invariance.
Specifically, let $\lambda_i (x,m)$ be an eigenvalue of the Hessian $H$.
The smooth under-approximation to \cref{eq:eqn_h} which ensures that the set invariance condition is achieved is 
$h^\circ(x) = -\frac{1}{\kappa} \ln(\sum_{i\in\mathbb{Z}} \text{exp}(-\kappa\,(\lambda_i(x,m) - \lambda_s))) \leq \min_{i\in\mathbb{Z}} \, \lambda_i(x,m) - \lambda_s =  h(x)$, $\kappa \in \mathbb{R}_+$.
Hence, replacing $h$ with $h^\circ$ in the set invariance condition stated in \cref{def:CBF} will ensure the minimum eigenvalue of the Hessian never becomes less than (or, as mentioned in \cref{remark:detect} greater than) the specified threshold $\lambda_s$.
To summarize, we smooth the barrier function $h$ from \cref{eq:eqn_h} to obtain a differentiable alternative $h^\circ$ to circumvent the non-differentiability of the min operator.
We then apply \cref{eqn:smooth_relu} to obtain an analytic control input so that the eigenvalues of the Hessian are themselves analytic and therefore differentiable.

\subsection{Anti-Crossing Control Barrier Function}
Requiring the dynamics and controller to be analytic is needed for the eigenvalues to remain differentiable if the minimum eigenvalue is non-simple. 
However, if it can be shown the eigenvalues are always simple, i.e., never cross, then the analytic requirement can be eliminated and the differentiability of the eigenvalues is guaranteed. 
Furthermore, the initial minimum eigenvalue will always be the minimum eigenvalue so the min operator can be removed from \cref{eq:eqn_h} and the smoothing under-approximation modification outlined in the previous subsection is no longer necessary. 
It is therefore desirable to prevent eigenvalue crossing.
We now introduce an additional CBF that works in conjunction with an I-CBF to prevent the eigenvalues from crossing.
Fundamentally, the anti-crossing CBF maintains a small gap between each pair of eigenvalues ensuring the difference between them is at least ${\delta^{\times}}\in \mathbb{R}_+$.
This can be imposed by defining the CBF
\begin{equation}
\label{eq:h_anti}
    h_i^{\times}(x)=\lambda_{i+1}(x,m)-\lambda_i(x,m) - {\delta^{\times}}, ~~ i=1, \dots, n-1.
\end{equation}
Since the eigenvalues are now required to not cross, we have $\lambda_{i+1}>\lambda_i$ and so $h_i^{\times}$ is positive when $\lambda_{i+1}-\lambda_i > {\delta^{\times}}$ thereby defining the desired super-level set. 
There are several alternatives to \cref{eq:h_anti} but the advantage of our choice is its simplicity and effectiveness (as was also the case for \cref{eq:eqn_h}).

To show \cref{eq:h_anti} is minimally invasive, we leverage the following result that provides insights about the structure of the manifold on which the eigenvalues could cross.

\begin{proposition}
    The eigenvalues of an $n\times n$ real symmetric matrix can only cross on a manifold of dimension $n-2$.
       \label{prop:manifold}
\end{proposition}
\begin{proof}
    A detailed proof can be found in \cite{magnus2019matrix}, but the core idea of the proof is based on the comparison between the number of independent parameters in a symmetric matrix and the number of parameters in its diagonalized representation. 
\end{proof}

The anti-crossing CBFs $h_i^\times$ are minimally invasive as eigenvalues of the Hessian can only cross on a manifold of size $n-2$ where $n$ is the size of the Hessian as stated in \cref{prop:manifold}.
This property implies that the system's trajectory will avoid this manifold without significantly deviating from the desired control input $u_d$ or the desired state $x_d$, as it can {circumvent} the lower dimensional manifold {with minimal change}.
For instance, if $n = 3$, the points to avoid are one-dimensional, meaning that the trajectory must avoid specific lines in three-dimensional space.

Regarding implementation, the analytic method requires stronger assumptions (Assumption \ref{assump_new}) but can be computationally faster as seen in \cref{sec:example} while the anti-crossing method is more intuitive and uses fewer hyperparameters.

\subsection{Predictive Measurement Models}
\label{sub:predictive}
    A subtle but important element of the proposed approach is the need to have information regarding the rate of change of measurements with respect to time. 
    In other words, the approach utilizes a form of measurement prediction to achieve set invariance in a similar manner to how the system's dynamics are used.
    For the sake of clarity, let $\lambda_{\text{min}}$ be simple. 
    Then, from \cref{eq:eqn_h} and \cref{def:icbf}, $\dot{h}(x,m) = \frac{\partial h}{\partial x} \dot{x} + \frac{\partial h}{\partial m} \dot{m} \geq - \alpha(h(x,m))$
    where the time variation of the measurements $m$ must be accounted for to achieve set invariance.
    Intuitively, the derivative can be thought of as providing an instantaneous prediction of the next measurement.
    In practice, one rarely has access to $\dot{m}$, so it must be estimated directly or approximated via $\mathfrak{m}(x)$.
    Future works will investigate this trade-off, but we expect the latter---which we call using \emph{predictive measurement models}---will be a fruitful research area given the various sensing modalities, such as vision, LiDAR, event camera, etc., used on mobile systems.  
    This work will utilize the approximation $m \approx \mathfrak{m}(x)$.

\subsection{Higher Relative Degree I-CBFs}
The previous analysis can be extended to control barrier functions that have a relative degree greater than one. 
Various approaches for constructing CBFs suitable for higher relative-degree systems have been suggested in the literature \cite{nguyen2016exponential,xiao2019control,cohen2024constructive}. 
Among these, one method that works well for I-CBFs is that presented in \cite{xiao2019control}.
In this method, a new CBF $h_r$ is constructed with a relative degree of one that in the limit has the same safe set as the original $h$.
This new $h_r$ is then an implementable I-CBF.
In the previous analysis, if the relative degree for the original $h$ is $r>1$, Assumption \ref{assump_mm} needs to be modified to $r$-times continuously differentiable, to ensure that the eigenvalues are sufficiently smooth. 
This method for higher relative degrees works for both the analytic Hessian and anti-crossing CBF, as is demonstrated in \cref{sec:example}.

\section{Illustrative Example}
\label{sec:example}

The proposed method is demonstrated on a 2-D double integrator system that uses beacons at known locations for localization. 
{This system} has a relative degree of two so we employ \cite{xiao2019control} to achieve set invariance.
Specifically for the analytic method, we use CBF $h_r^\circ (x)=h^\circ(x) ( 1+ \sigma(L_f^{r-1} h^\circ(x)))-\delta$ where $\sigma : \mathbb{R} \rightarrow \mathbb{R} : \zeta \mapsto (1 + \exp({-\zeta}))^{-1}$ and $\delta \in \mathbb{R}_{+}$.
For the anti-crossing method we use $h_r(x)=h(x) ( 1+ \sigma(L_f^{r-1} h(x)))-\delta$ and $h_r^\times(x)=h^\times(x) ( 1+ \sigma(L_f^{r-1} h^\times(x)))-\delta$.
At each time step the system's state is computed by solving a nonlinear least squares problem via gradient descent, {with either} range-only {or} bearing-only measurements as these are common in practice.
The estimated state and the Hessian of the nonlinear least squares cost are then used by the safety controller to override an LQR controller that is trying to drive the system to a desired location.
The gradient of the minimum eigenvalue was computed analytically.
Per \cref{sub:predictive}, the measurement model $\mathfrak{m}(x)$ was used to predict the rate of change of the measurements. 
The position of the system is $(p_x,\,p_y)$ with beacons located at $({b^k_{x}},\,b^k_{y})$ for $k=1,2,3$.

\subsection{Range-Only Measurements}
The measurement model {for range-only beacons} is $\mathfrak{m}_k(p_x ,p_y )=\sqrt{(p_x-b^k_{x})^2+(p_y-b^k_{y})^2}$ for $k=1,2,3$.
The covariance matrix $\Sigma(p_x,p_y)$ for each measurement was chosen to be $\Sigma(p_x,p_y)_{k,k}={1+\exp({m_k(p_x,p_y)-10})}$ to capture measurement dropout due to distance from the $k^{\text{th}}$ beacon. 
The off-diagonal entries were set to zero.

\textit{Localization.}
For localization, the safe region where $\lambda_s\geq5$ is shown as the region inside the red line in \cref{fig:distance_trajectories} with the grey scale that denotes the value of $\lambda_{\text{min}}$ as a function of $p_x$ and $p_y$.
Both the analytic and anti-crossing methods demonstrated good performance in \cref{fig:distance_barrier}, ensuring that $h$ remained positive and that the minimum eigenvalue stayed above the specified threshold.
The anti-crossing method was more aggressive as the system took a more direct path to the goal and subsequently got closer to the boundary of $\mathcal{S}$.
Moreover, the max control effort was twice as large in this method compared with the analytic I-CBF, highlighting the trade-off between the analytic and anti-crossing formulations.
The average per-step computation time was 0.7 ms for the analytic method and 35.6 ms for the anti-crossing method.
The majority of the computation time for the anti-crossing method came from the optimization while formulating the Hessian and its derivatives took only 0.08 ms.

\textit{Detection Avoidance.} As noted in \cref{remark:detect} the method in this letter can also be used for detection avoidance by ensuring that the minimum eigenvalue of the hessian remains sufficiently small. The safe set where detection can be avoided is $\lambda_s\leq5$ and is indicated by the red line in \cref{fig:distance_trajectories_avoid}. The system remained outside of the detectable set as the value of the barrier function remained positive \cref{fig:distance_barrier_avoid}. Like with the localization, the anti-crossing method was more aggressive as compared to the analytic method.

\begin{figure*}[t!]
\centering
  \subfloat[Localization trajectory.]{\includegraphics[trim=0 40 0 50, clip, width=.43\textwidth]{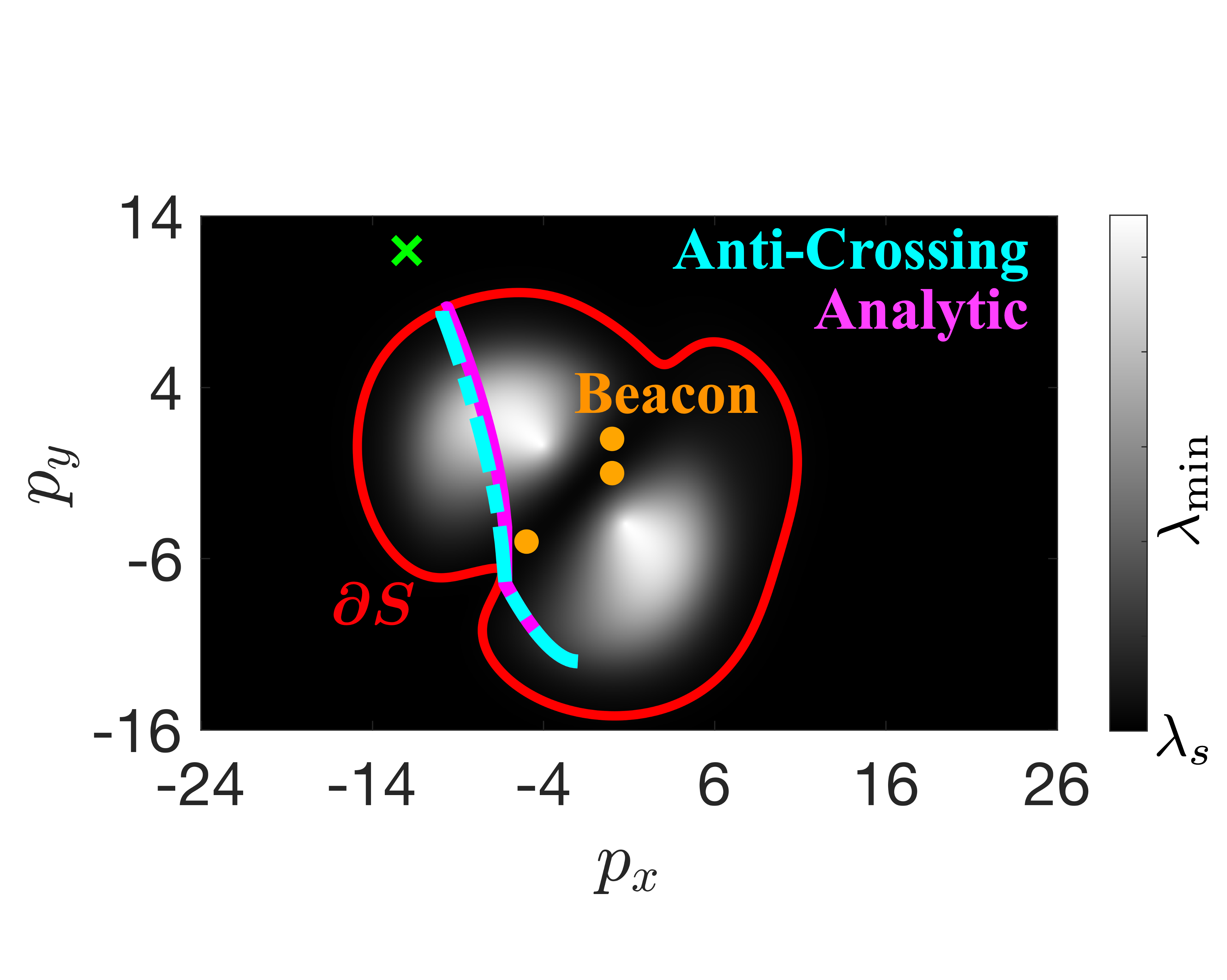}\label{fig:distance_trajectories}}
  \hspace{0.2em}
  \subfloat[Detection avoidance trajectory.]{\includegraphics[trim=0 40 0 50, clip, width=.43\textwidth]{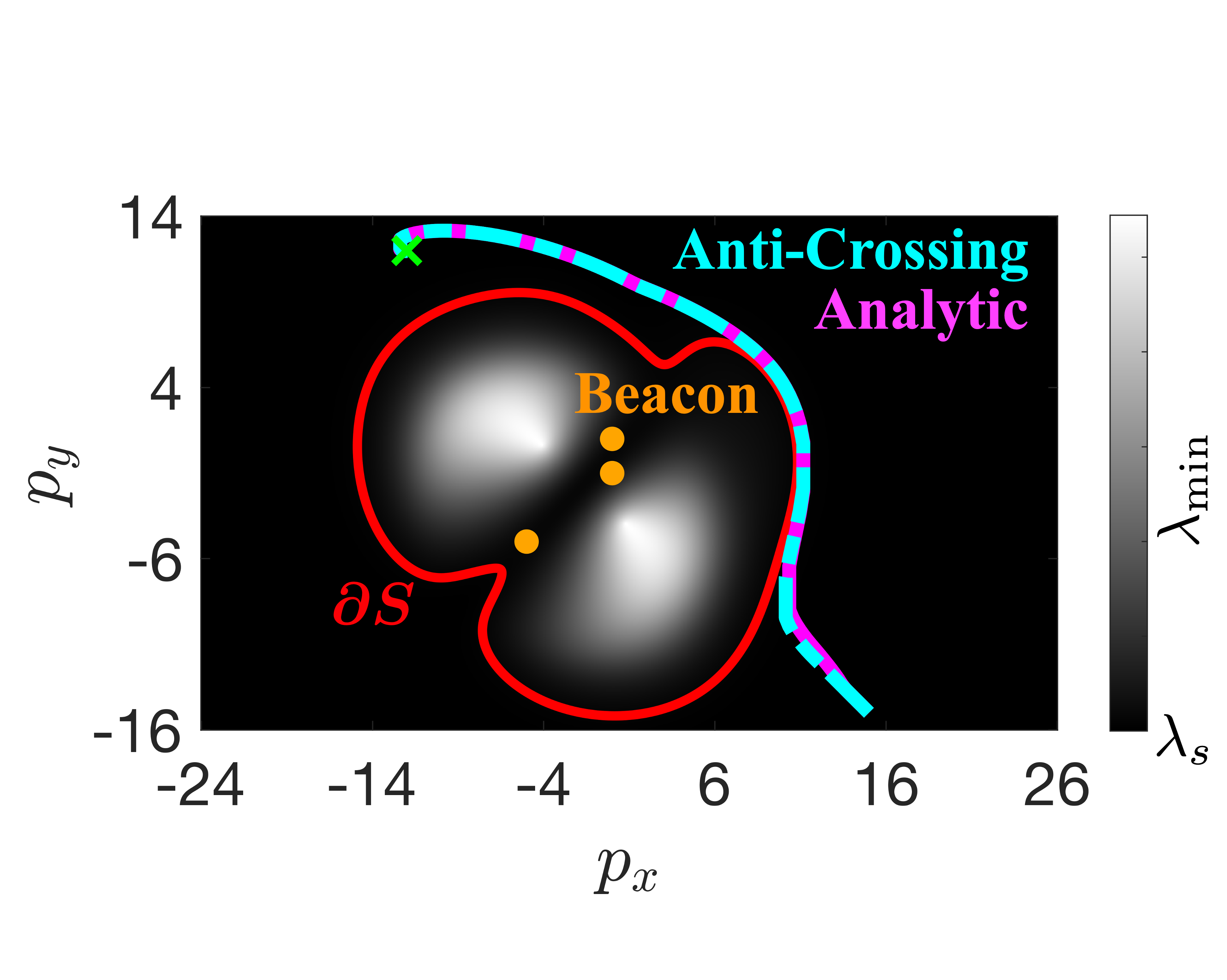}\label{fig:distance_trajectories_avoid}}
  \hspace{0.2em}
  \subfloat[Localization barrier value.]{\includegraphics[trim=0 5 0 0, clip, width=.43\textwidth]{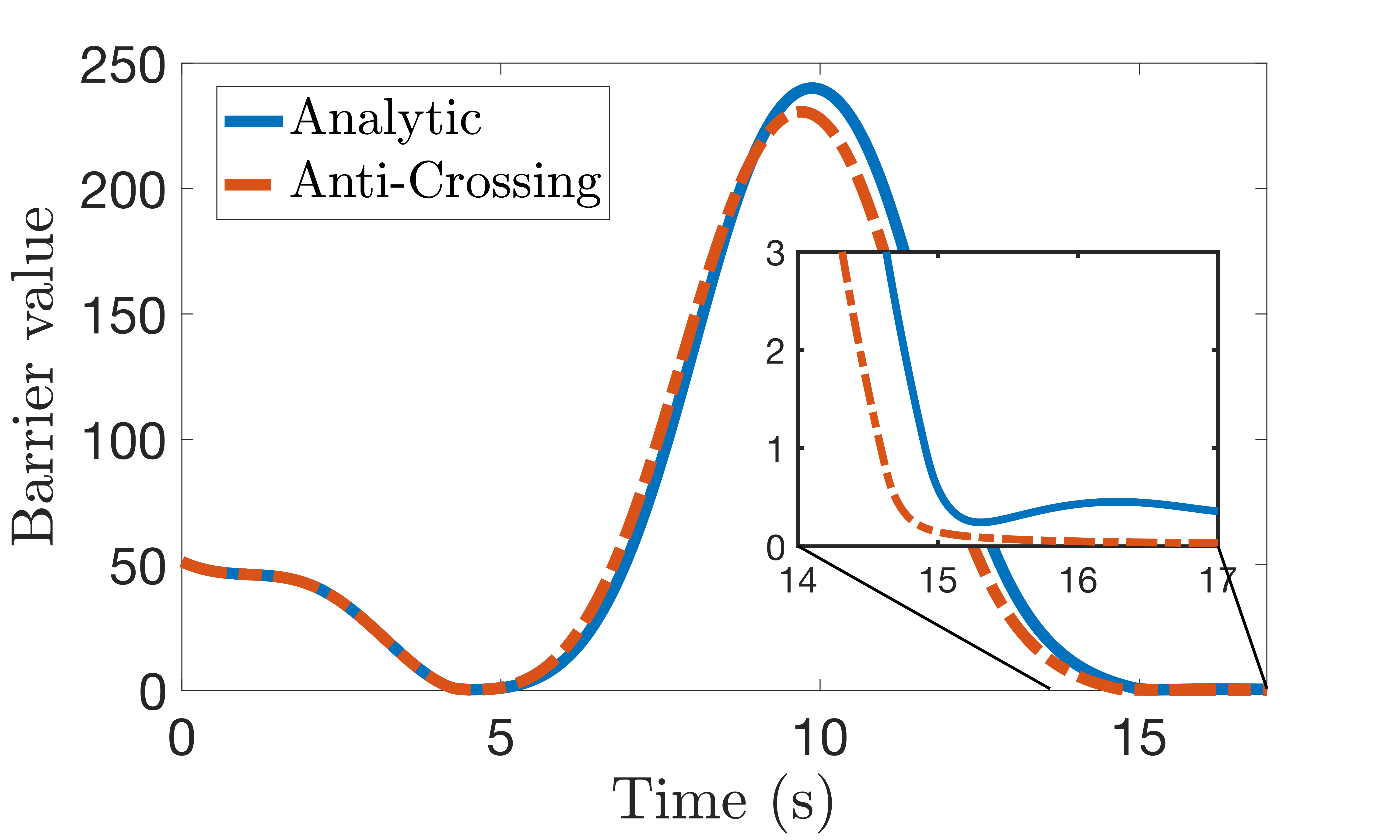}\label{fig:distance_barrier}}
  \hspace{0.2em}
  \subfloat[\centering Detection avoidance barrier value.]
  {\includegraphics[trim=0 5 0 0, clip, width=.43\textwidth]{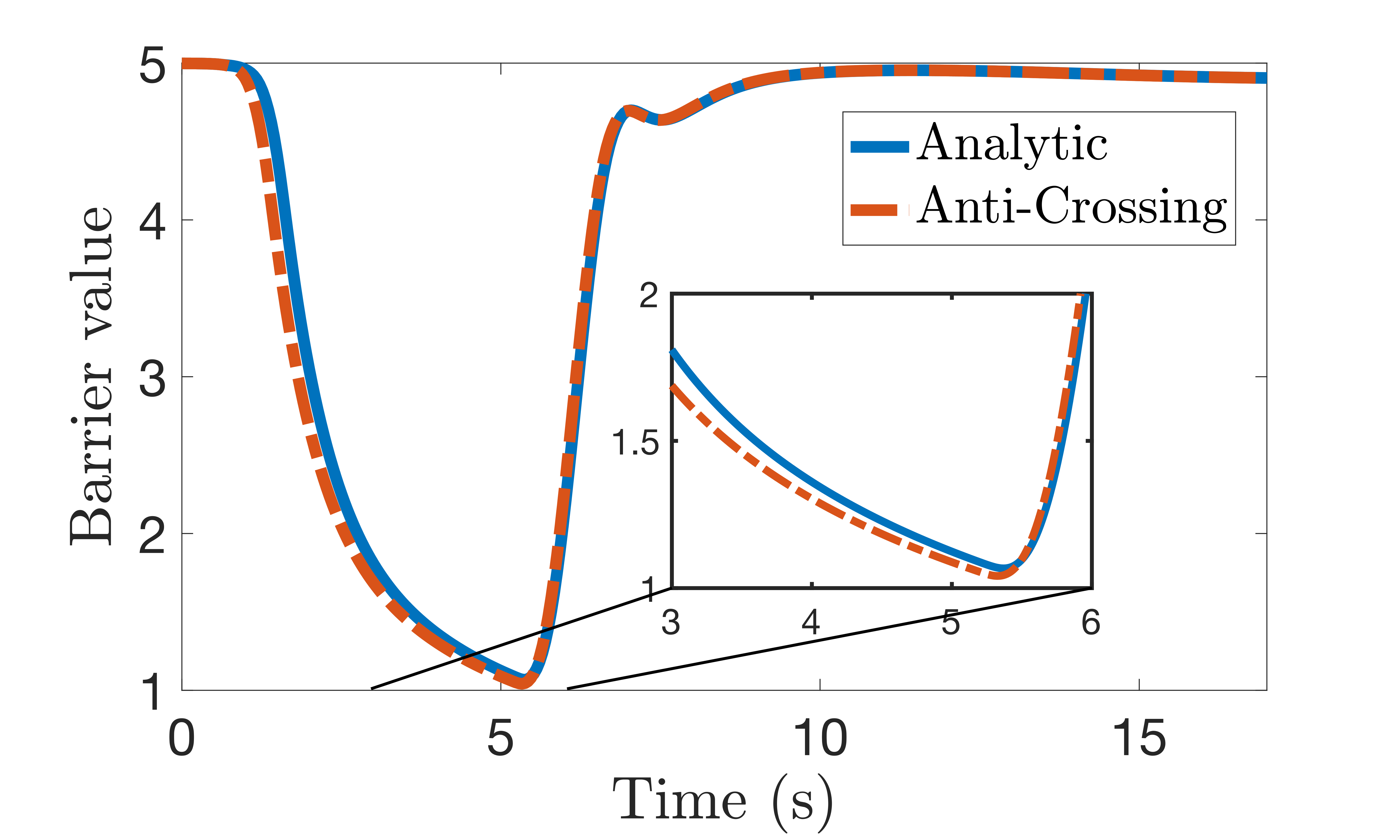}\label{fig:distance_barrier_avoid}}
  \caption{Trajectories and barrier values for range-only measurements. In the trajectory plots \cref{fig:distance_trajectories,fig:distance_trajectories_avoid} the safe region is the interior of the closed red curves and the desired state with the green $\times$. The analytic and anti-crossing methods prevent the system from leaving the safe set, but the analytic is more conservative as indicated by the larger value of the barrier function.}
  \label{fig:distance_results}
  \vskip -0.2in
\end{figure*}

\subsection{Bearing-Only Measurements}
Next, we demonstrate the method using bearing-only measurements for localization. 
The measurement model is
$\mathfrak{m}_k(p_x ,p_y )=\tan^{-1}\left((p_y-b^k_y)/(p_x-b^k_x)\right)$ for $k=1,2,3$ and the weighting matrix is $\Sigma(p_x,p_y)=I$ since the eigenvalues naturally decayed as the distance from the beacons increased thereby eliminating the need to model measurement degradation/dropout separately.

\textit{Localization.}
For safe localization the safe region, where $\lambda_s\geq0.01$ is shown in \cref{fig:angle_trajectory} by the red line.
Looking at \cref{fig:angle_barrier}, both the analytic and anti-crossing methods maintained $h$ above zero, but their values deviated substantially. 
As with the range-only measurements, the analytic method was more conservative, but the peak control effort was approximately five times less.
The average per-step computation time for the analytic and anti-crossing methods was 0.8 ms and 19.4 ms respectively, in the anti-crossing computing the Hessian and its derivatives took only 0.06 ms.
Both methods ensured the eigenvalue of the Hessian did not drop below the specified threshold, even with the complex shape of the safe set $S$.

\textit{Detection Avoidance.}
The bearing-only measurements can also be used in a detection avoidance formulation. For detection avoidance, the safe region is where $\lambda_s\leq0.01$ is shown in \cref{fig:angle_trajectory_avoid} by the red line.
Both the anti-crossing method and the analytic methods kept the system in the safe set as indicated by the value of the barrier function \cref{fig:angle_barrier_avoid} with the anti-crossing exhibiting more aggressive behavior.

\begin{figure*}[t!]
\centering
  \subfloat[Localization trajectory.]{\includegraphics[trim=0 40 0 50, clip, width=.43\textwidth]{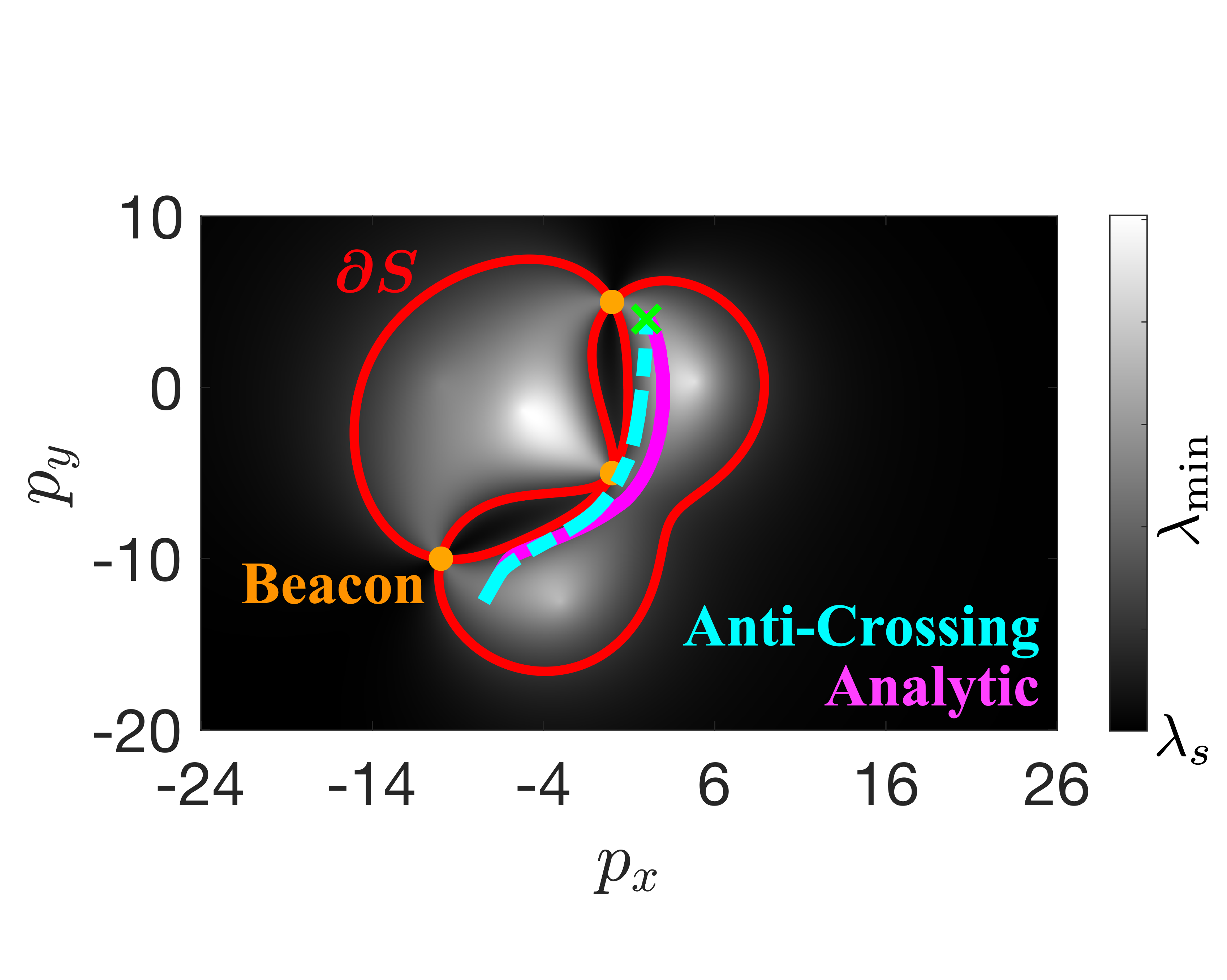}\label{fig:angle_trajectory}}
  \hspace{0.2em}
  \subfloat[Detection avoidance trajectory.]{\includegraphics[trim=0 40 0 50, clip, width=.43\textwidth]{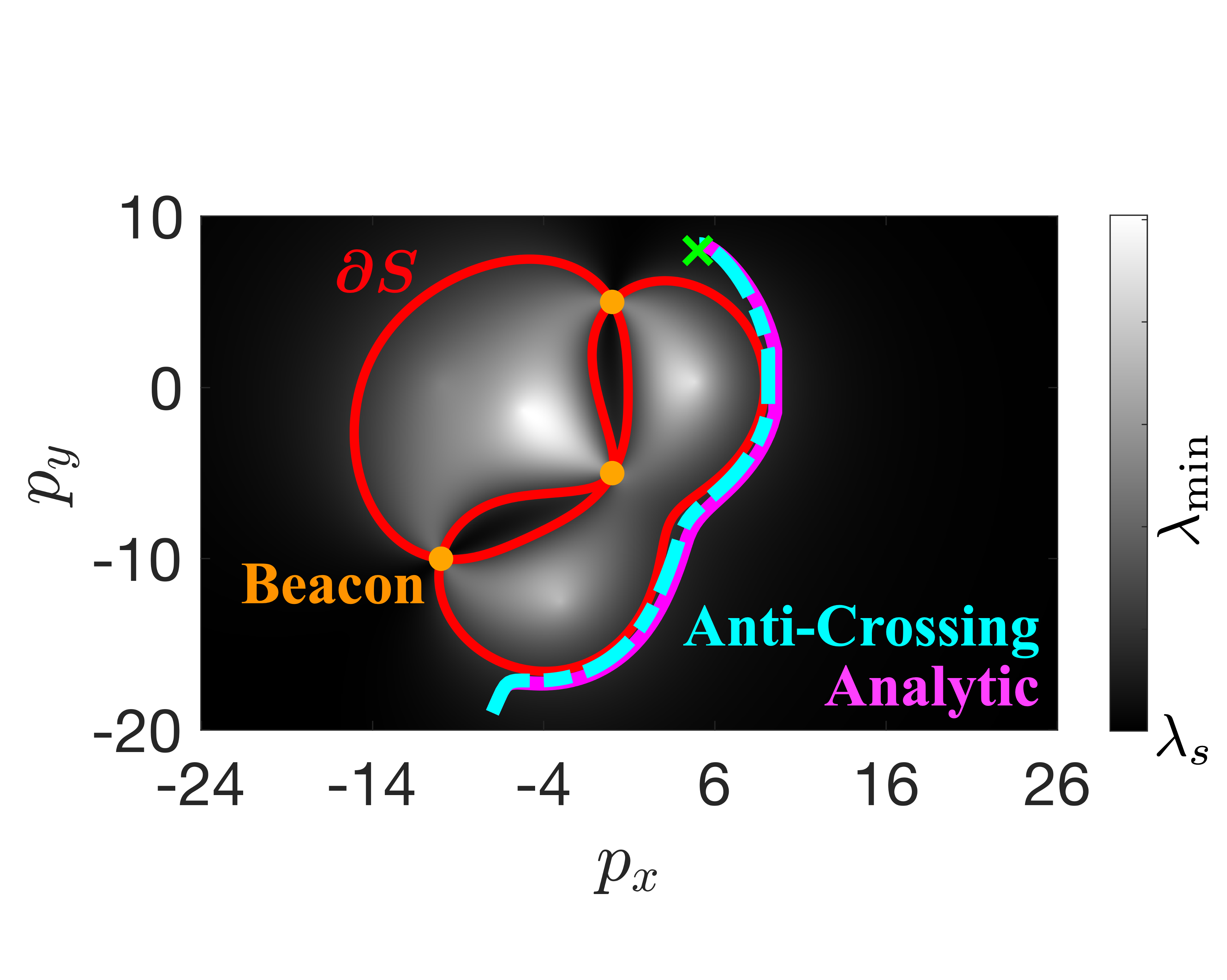}\label{fig:angle_trajectory_avoid}}
  \hspace{0.2em}
  \subfloat[Localization barrier value.]{\includegraphics[trim=0 5 0 0, clip, width=.43\textwidth]{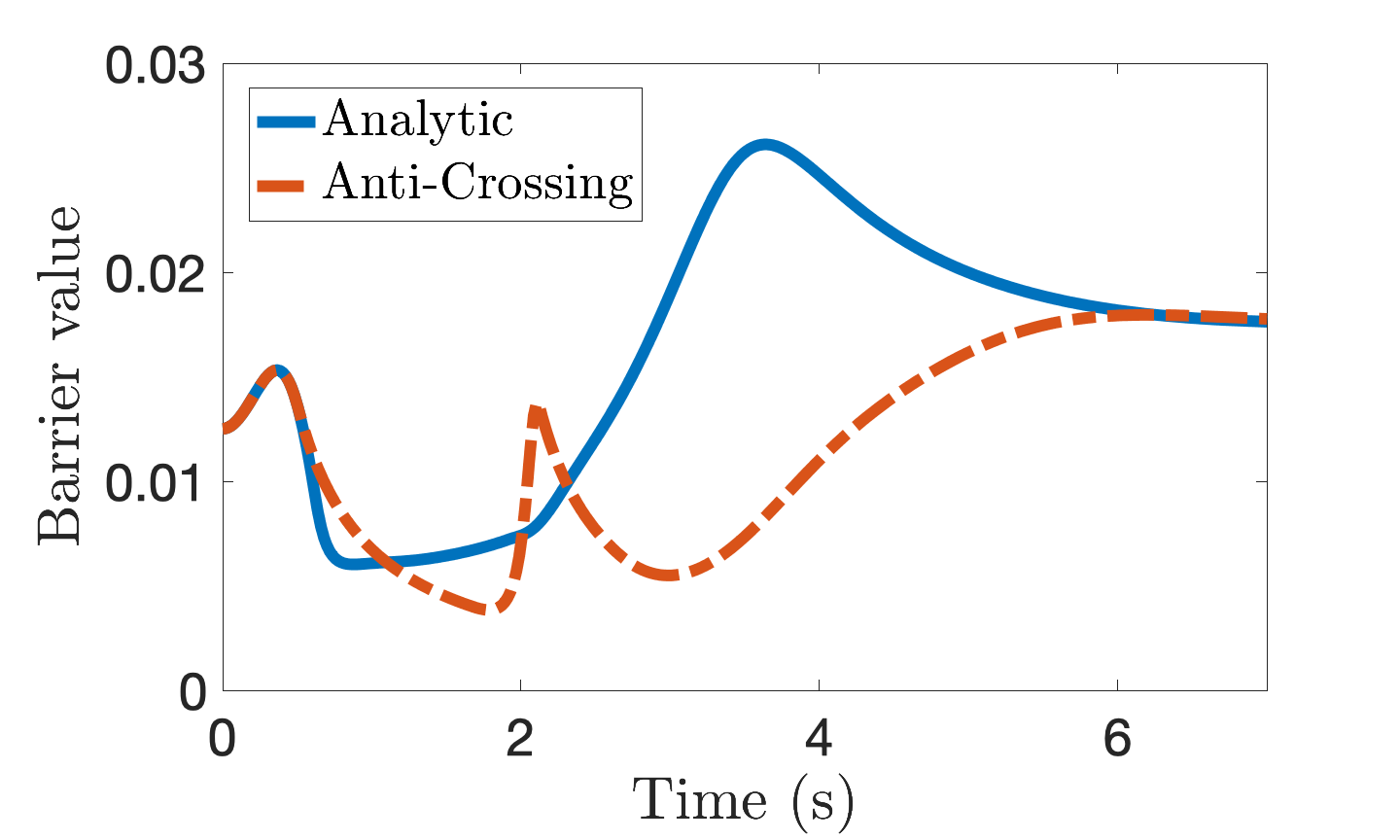}\label{fig:angle_barrier}}
  \hspace{0.2em}
  \subfloat[\centering Detection avoidance barrier value.]
  {\includegraphics[trim=0 5 0 0, clip, width=.43\textwidth]{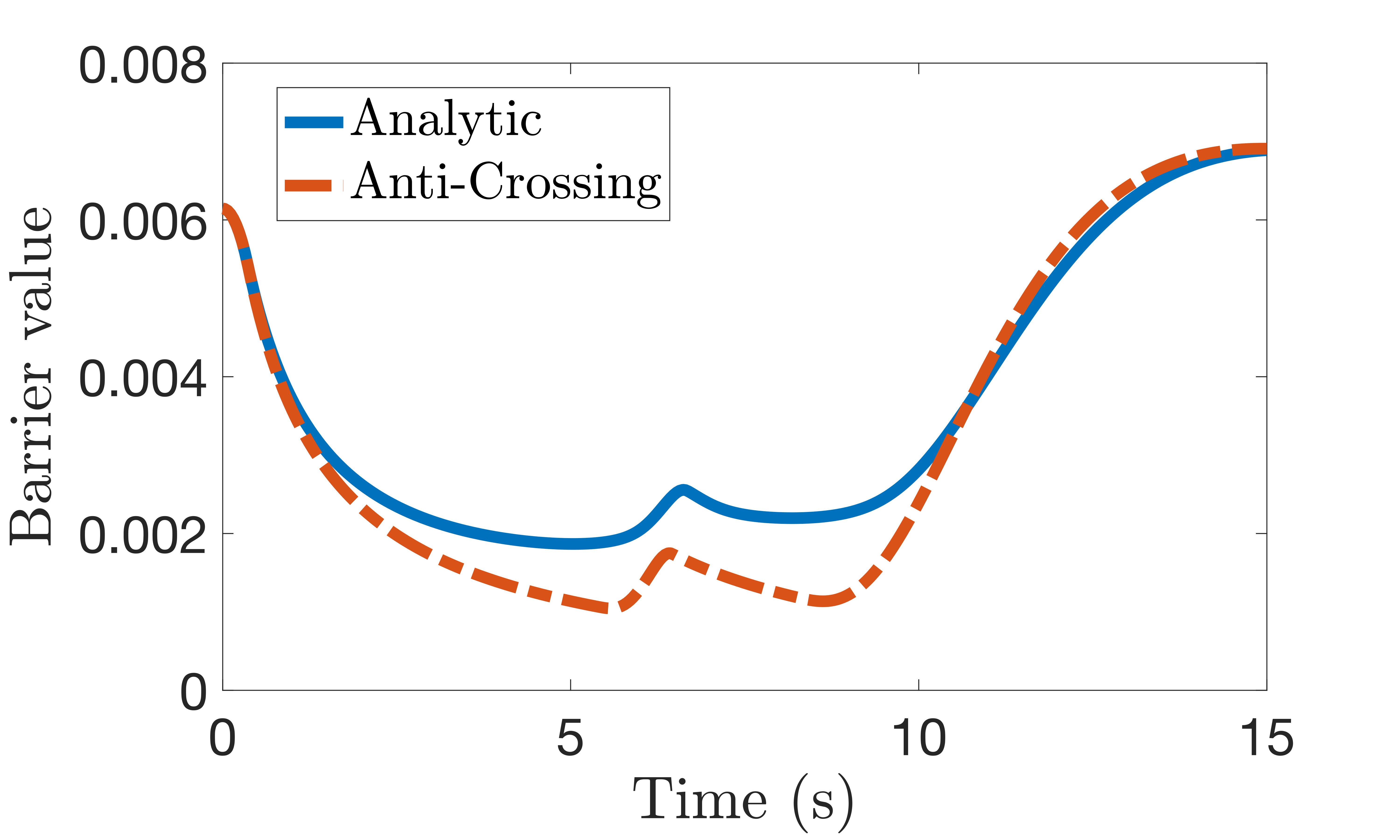}\label{fig:angle_barrier_avoid}}
  \caption{Trajectories and barrier values for bearing-only measurements. In the trajectory plots \cref{fig:angle_trajectory,fig:angle_trajectory_avoid} the safe region is the interior of the closed red curves and the desired state with the green $\times$. The analytic and anti-crossing method prevent the system from leaving the safe set, but the analytic is more conservative as indicated by the larger value of the barrier function.}
  \label{fig:angle_results}
  \vskip -0.2in
\end{figure*}
\section{Conclusion}
This letter presented a method for maintaining reliable localization in mobile systems by leveraging Information Control Barrier Functions (I-CBFs) to impose constraints on the eigenvalues of the Hessian in optimization-based localization algorithms.
Unlike previous approaches, which often require prior knowledge of the environment or are tailored to specific situations, this framework offers a general solution applicable to various systems. 
Future work will focus on applying this method to vision and LiDAR-based localization \cite{chen2023dliom} and exploring how some of the assumptions in this letter might be relaxed. Additionally, integrating robust, stochastic, and adaptive methods \cite{lopez2023unmatched} could broaden its applicability and {provide} safety guarantees for uncertain systems. 
The method was demonstrated through simulations on range-only and heading-only beacon localization, showing that both proposed methods preserve localization safety.
\label{sec:conclusion}

\section{Appendix}
Here we include derivations and formulas for computing the derivatives of an simple eigenvalue of a symmetric matrix for completeness.  
If $\lambda$ and $v$ are the eigenvalue and (normalized) eigenvector of a time-dependent matrix $A : \mathbb{R} \rightarrow \mathbb{R}^{n\times n}$ then by definition $A v = \lambda v$ which can be differentiated to get $\dot A v + A \dot v = \dot \lambda v + \lambda \dot v$.
Multiplying both sides by $v^\top$ one obtains $\dot \lambda = v^\top \dot A v$.
Using this relation, the first derivative of $v$ is $\dot v = (A-\lambda I)^{\dagger}(\dot \lambda I - \dot A) v,$
where $\dagger$ is the psuedoinverse. 
Differentiating $\dot{\lambda}$ and $\dot{v}$ again, $\ddot{\lambda} = \dot{v}^T \dot{A} v + v^T \ddot{A} v + v^T \dot{A} \dot{v}$ and $\ddot{v} = (A - \lambda I)^{\dagger} ( \ddot{\lambda} I - \ddot{A} ) v + (A - \lambda I)^{\dagger} ( \dot{\lambda} I - \dot{A} ) \dot{v}.$

Simulation parameters are listed in \cref{table:sim1,table:sim2,table:sim3,table:sim4}.
Unless otherwise noted the simulation time was $dt=1\times 10^{-4}$ s and all initial and final velocities were zero.

\begin{table}[h]
\caption{RANGE-ONLY LOCALIZATION PARAMETERS}
\label{table:sim1}
\vspace{-.8em}
\begin{center}
\begin{tabular}
{ c||c|c  }
 \hline
 & Analytic &Non-Coalescing\\
 \hline
 $\alpha(\cdot)$ & $10h_r^\circ(x)^2$     &$10h_r(x)^2$, \ $100 h_r^\times(x)^2$ \\
 $\delta$ (m)&   $0.01$  & $0.01$  \\
  $\delta^\times$ (m)&   N/A  & $0.01$  \\
 $c$&   $1$  & N/A  \\
 $\kappa$&   $1$  & N/A  \\
 $v_x(0)$ (ms$^{-1}$) & -1 & -1 \\
 \hline
\end{tabular}
\end{center}
\vspace{-0.25in}
\end{table}

\begin{table}[h]
\caption{RANGE-ONLY DETECTION AVOIDANCE PARAMETERS}
\label{table:sim2}
\vspace{-.8em}
\begin{center}
\begin{tabular}
{ c||c|c  }
\hline
 & Analytic &Non-Coalescing\\
 \hline
 $\alpha(\cdot)$ & $0.1 h_r^\circ(x)^2$     &$0.1 h_r(x)^2$, \ $500 h_r^\times(x)^2$ \\
 $\delta$ (m)&   $0.01$  & $0.01$  \\
   $\delta^\times$ (m)&   N/A  & $0.01$  \\
 $c$&   $10$  & N/A  \\
 $\kappa$&   $10$  & N/A  \\
 \hline
\end{tabular}
\end{center}
\vspace{-0.25in}
\end{table}

\begin{table}[h]
\caption{BEARING-ONLY LOCALIZATION PARAMETERS}
\label{table:sim3}
\vspace{-.8em}
\begin{center}
\begin{tabular}
{ c||c|c  }
 \hline
 & Analytic &Non-Coalescing\\
 \hline
 $\alpha(\cdot)$ & $500h_r^\circ(x)^2$     &$100h_r(x)^2$, \ $10 h_r^\times(x)^2$ \\
 $\delta$ (m)&   $1\times 10^{-6}$  & $1\times 10^{-6}$  \\
   $\delta^\times$ (m)&   N/A  & $0.01$  \\
  $c$&   $100$  & N/A  \\
   $\kappa$&   $1000$  & N/A  \\
 \hline
\end{tabular}
\end{center}
\vspace{-0.25in}
\end{table}

\begin{table}[h]
\caption{BEARING-ONLY DETECTION AVOIDANCE PARAMETERS}
\label{table:sim4}
\vspace{-.8em}
\begin{center}
\begin{tabular}
{ c||c|c  }
 \hline
 & Analytic &Non-Coalescing\\
 \hline
 $\alpha(\cdot)$ & $100h_r^\circ(x)^2$     &$100h_r(x)^2$, \ $500 h_r^\times(x)^2$ \\
 $\delta$ (m)&   $1\times 10^{-6}$  & $1\times 10^{-6}$  \\
   $\delta^\times$ (m)&   N/A  & $0.01$  \\
  $c$&   $5000$  & N/A  \\
   $\kappa$&   $5000$  & N/A  \\
 \hline
\end{tabular}
\end{center}
\vspace{-0.3in}
\end{table}

\label{sec:appendix}

\FloatBarrier
\bibliographystyle{IEEEtran}
\bibliography{references.bib}

\end{document}